\documentclass[a4paper,11pt]{article}
\usepackage{paralist}
\usepackage{amsfonts}
\usepackage{xspace}
\usepackage{verbatim}
\usepackage{amssymb, amsmath, boxedminipage}
\usepackage[margin=1.0in]{geometry}
\usepackage{graphicx}
\usepackage{algorithm}
\usepackage{algorithmic}
\usepackage{amsthm}
\usepackage{setspace}
\usepackage{color}
\usepackage{enumitem}

\newcommand{\XOR}{\bigoplus}

\newcommand{\Find}{{FindMin}}
\newcommand{\anyfind}{{FindAny}}
\newcommand{\cP}{{\cal P}}
\newcommand{\Zp}{{\mathbb Z}_p}
\newcommand{\Sample}{{Sample}}

\newcommand{\HPTestOut}{{HP-TestOut}}

\newcommand{\drop}[1]{}
\newcommand{\mikkeldrop}[1]{}

\newtheorem {lemma} {Lemma}
\newtheorem {observation} {Observation}

\newcommand\eps\varepsilon
\newcommand\fct\rightarrow

\newcommand\ceil[1]{\lceil {#1}\rceil}

\newcommand\ol\overline

\newcommand{\mikkel}[1]{\textcolor{green}{{{\bf Mikkel:} #1\\}}}
\newcommand{\val}[1]{\textcolor{blue}{{{\bf Val:} #1\\}}}
\newcommand{\old}[1]{}

\newcommand{\TestOut}{{TestOut}}

\newtheorem{theorem}{Theorem}[section]

\begin{document}
\begin{titlepage}
\title{ Construction and impromptu repair of an MST in a distributed network with $o(m)$ communication}

\author{
Valerie King\thanks{University of Victoria; research supported by Simons Institute for the Theory of Computing, Berkeley; Institute for Advanced Study, Princeton,  ENS Paris, and NSERC.
 Research supported in parts by a
Google Faculty Research Award
 \hbox{E-mail}:~{\tt  val@uvic.ca}.}
\and
Shay Kutten\thanks{Faculty of IE\&M, Technion, Haifa 32000. Research supported in parts by the Israel Science Foundation, by the Ministry of Science, and by the Technion TASP center. \hbox{E-mail}:~{\tt kutten@ie.technion.ac.il}.}
 \and
Mikkel Thorup\thanks{University of Copenhagen, Denmark.  \hbox{E-mail}:~{\tt   mikkel2thorup@gmail.com}.}
}

\date{}
\maketitle

\setcounter{page}0
\thispagestyle{empty}

\begin{abstract}
In the CONGEST model, a communications network is an undirected graph whose $n$ nodes are processors and whose $m$ edges are the communications links between processors. At any given time step, a message of size $O(\log n)$ may be sent by each node to each of its neighbours. We show for the synchronous model:
If all nodes start in the same round, and each node knows its ID and the ID's of its neighbors, or in the case of MST, the distinct weights of its incident edges and knows $n$, then there are Monte Carlo algorithms which succeed w.h.p. to determine a minimum spanning forest (MST) and
a spanning forest (ST) using $O(n \log^2 n/\log\log n)$ messages for MST  and $O(n \log n )$ messages for ST, resp.  These results contradict the ``folk theorem" noted in Awerbuch, et.al., JACM 1990 that the distributed construction of a broadcast tree requires $\Omega(m)$ messages. This lower bound has been shown there and in other papers for
some CONGEST models; our protocol demonstrates the limits of these models.

A dynamic distributed network is one which undergoes online edge
insertions or deletions.  We also show how to repair an MST or ST in a
dynamic network with asynchronous communication. An edge deletion can
be processed in $O(n\log n /\log \log n)$ expected messages in the MST,
and
$O(n)$ expected messages for the ST problem, while an edge insertion
uses $O(n)$ messages in the worst case. We call this ``impromptu" updating as we assume
that between processing of edge updates there is no preprocessing or
storage of additional information.  Previous algorithms for this
problem that use an amortized $o(m)$ messages per update require
substantial preprocessing and additional local storage between
updates.

\end{abstract}

\end{titlepage}

\section{Introduction}

The problem of finding a minimum spanning forest (MST)  or computing a spanning forest (ST) in a communications network  is one of the most fundamental and heavily studied problems in distributed computing.
This problem is important for facilitating broadcast and coordination in a message-efficient manner. Given such a tree, messages may be broadcast from one node to all others or values from all nodes can be combined from the leaves up to one node in time proportional to the diameter of the tree, with a number of messages which is
proportional to the size of the tree, rather than all edges in the network, as when communication is by flooding.  For this reason, a tree is useful for tasks such as leader election, mutual exclusion, and reset (adaptation of any static algorithm to changes in the network topology). Below, we consider a network to be a graph with $n$ nodes and $m$ edges.

In 1983, Gallager, Humlet and Spira gave a now classic algorithm for finding a MST in a distributed asynchronous communications network with message complexity $O(m+n\log n)$ for a network with $n$ nodes and $m$ edges. The message complexity for this problem has not been improved until now, even for the easier context we consider here (a synchronous network with nodes initialized to start the algorithm at the same time). For the unweighted problem of ST,
a single node starting a flooding algorithm can can construct a broadcast tree $O(m)$ messages in time equal to the diameter of the network (see, e.g. \cite{segal}).

That $\Omega(m)$  messages are required for broadcast (and the construction of an ST)  is mentioned as ``folklore"
 by
Awerbuch, Goldreich, Peleg and Vainish (1990) \cite{vainish}. 
They prove this lower bound
 in what is referred to as the ``standard" $KT_1$ model, where each node knows
  (its own identity) and the identity of its neighbors. The $\Omega(m)$ lower bound holds for randomized (Monte Carlo) comparison protocols, where the basic computation step is to compare two processors' identities, and for general algorithms
where the set of ID's is very large and grows independently with respect to message size, time and randomness. In 2013  Kutten, Pandurangan, Peleg, Robinson, and Trehan showed an $\Omega(m)$ lower bound for randomized general algorithms in the $KT_0$ model, where each node does not know the identities of its neighbors \cite{Kutten2013}.  All these lower bounds hold when the size of the network is known to all the nodes, the network is synchronous, and all the nodes start
Simultaneously. Our MST and ST algorithms avoid both lower bounds by assuming $KT_1$ and an
exponential bound on the size of the identity space.


\smallskip

Communication networks are inherently dynamic, in that a link may be either deleted or inserted over time.
This paper
also presents
 algorithms to repair an MST or ST in an asynchronous network upon an edge insertion or deletion.
  These algorithms have the
   new property (for an efficient dynamic graph algorithm) of being ``impromptu", that is, they require no preprocessing or storage of auxiliary information except during the processing of the current updates. Between updates, a node knows only the names and weights of its incident edges and whether these edges are in the currently maintained MST or ST.
While there are previously known algorithms for updating MST and ST with $O(n)$ messages, these have significant memory requirements and require the communication costs to be amortized over a sequence of sufficiently long updates. For example, the 2008 algorithm of Awerbuch et. al. \cite{awerbuch-2008}  to maintain an MST uses $O(n)$ amortized messages per update (somewhat better then our second algorithm), but stores and stores  $\Theta(\Delta_v n \log n )$  bits at each node $v$, where $\Delta_v$ is the number of node $v$'s neighbors \footnote{
According to \cite{awerbuch-2008}, ``keeping
track of history enables significant improvements in the communication
complexity of dynamic networks protocols.'' Atop its abstract,
this was also stated in a more lyrical way:
  ``Those who cannot remember the past are condemned to repeat it
  (George Santayana).'' A part of the message of the current paper may
  be adding ``unless they flip coins...'' Here we show
that history can be replaced by random coin tosses.}.

\medskip

We first describe the model and then the
results.  A
communications network is a graph. We assume that every node knows a bound $n$ on the actual number of nodes.
An interesting case is when the known upper bound  on the size of the network is very tight (e.g the actual size multiplied by some small positive  integer constant).
In this case, all our asymptotic results are in terms of the actual network size. Hence for simplicity, we refer below to $n$ as the network size (rather than an upper bound).
The communication links are
undirected edges and each node has a unique ID $\in \{1,2,..,n^c\}$. In fact,
using the classic Karp-Rabin \cite{karp-rabin} fingerprinting, w.h.p., we
can easily map $n$ ID's in exponential ID space to distinct ID's in polynomial ID space.
  For the MST problem, each edge 
has a weight
$\in \{1,2,...,u\}$ for any positive integer $u$.  Each node
knows its own ID, the weight of each incident edge, the ID of its
other endpoint, and $n$.
No other information
about the graph is known to any node at the start of any of our
algorithms.
 A message is a communication of $O(\log (n+u))$ bits which
is passed along a single edge.

A network is {\it properly marked} if every edge is marked by both or neither of its endpoints. {\it A tree $T$ is maintained by a network} if  the network is properly marked and $T$ is a maximal tree in the subgraph of marked edges.
For a node $x$, let
$T_x$ denote the tree maintained by the network and containing node $x$. We call an (unmarked) edge with exactly one endpoint in $T $ an edge
 {\it leaving}
  $T$ or  {\it outgoing}. A tree construction problem assumes that initially all edges are unmarked and every node
  knows to begin construction.
  At the end of the algorithm the network should maintain the MST (or ST).
We use the usual definitions of synchrony and asynchrony:
A {\it synchronized} network assumes a global clock, and messages are received in one time step.
An {\it asynchronous} network assumes that messages  are eventually received. Each node's action is triggered by the receiving of a message or other change to its state.
We say an event occurs ``w.h.p." (with high probability) if for
any constant $c$ (which is given as a parameter of the algorithm), the probability of the event is at least $1-n^{-c}$.
We show:

\begin{theorem} \label{t:build} There are algorithms to construct a minimum spanning tree (MST) and spanning tree (ST)
succeeding w.h.p. in a synchronous
networks of $n$ nodes using time and messages $O(n \log^2 n /\log \log n)$ for
MST and $O(n \log n)$ for
ST, and $O(\log (n+u))$ local memory per node. This assumes that each node is initialized to start the algorithm and  its only initial knowledge of the graph is its ID, its neighbors' ID's, the weight of each of its incident edge, and $n$.
  \end{theorem}

\begin{theorem} \label{t:repair}
Upon deletion or increase in weight of an edge,  there are algorithms \anyfind~ and \Find~ to repair an ST and an MST, respectively, which find a replacement edge if there is any, in an asynchronous distributed network using expeced time and messages $O(n)$ for the ST, and
$O(n\log n/\log \log n)$ for the MST, and $O(\log (n+u))$ local memory per node.  Upon insertion or decrease in weight of an edge, a deterministic algorithm with $O(n)$ time and messages suffices to repair the tree. All repairs are impromptu, i.e, no preprocessing or extra storage is needed between updates.
This assumes each node knows its ID and the ID's of their neighbors and the weight of each incident edge. To achieve success with
probability $1-n^{-c}$, each node must know an upper bound on $n$ which is within a polynomial of $n$.
\end{theorem}

Modified versions $\anyfind-c$ and  $\Find-c$ are also presented. Their
 worst case cost  matches the expected cost of $\anyfind$ and $\Find$.  When there is a replacement
for a deleted tree edge, w.h.p., they return either a correct replacement edge or $\emptyset$, and with constant probability, they return the former.

Our algorithms are based on the following new procedures which may be useful in other contexts.
Below, node $x$ initiates the procedure and receives the output, and $j,k \in \{1,..,u\}$:
\begin{itemize}[noitemsep]
\item
 \TestOut($x,j,k$): Returns true with constant probability if  there is an edge leaving $T_x$ with edge weight in the interval $[j, k]$; false otherwise. Always correct if true is returned.
 \item
 \HPTestOut($x,j,k$): The same as $\TestOut$ but w.h.p.
  \end{itemize}

 A basic  communication step in our network is a simple distributed routine  {\it broadcast-and-echo} \cite{GHS}.
 It is initiated by the broadcast of a message by a node $x$ which becomes the ``root" of a tree.
When a node $v$ receives a broadcast message from its neighbor $y$, it designates $y$ as its ``parent"  (for the sake of the current communication step) and sends a broadcast message to each of its other neighbors in $T$, its ``children''.
When a leaf node in $T$ receives a broadcast message, it sends a message (``echo") to its parent, possibly carrying some value. When a non-leaf message has received an echo message from every child, it sends an echo message to its parent, possibly aggregating its value with the values sent by its children. When the root has received echo messages from all its children, the broadcast-and-echo is done.
We show:

\begin{lemma} $\TestOut$ and $\HPTestOut$ can be performed with one broadcast-and-echo with message size $O(\log (n+u))$. The echo of $\TestOut$ requires only a message of only one bit.
 \end{lemma}

\noindent
{\bf Other previous work}

\noindent
{\it Similar techniques:} $\TestOut$ uses the principle that each edge with two endpoints in a tree contributes 0 to the parity of the sum of the degrees of the
nodes in a tree, while each edge which leaves a tree contributes 1. Therefore, in a randomly sampled graph, there is a 1/2 chance that if there are one or more edges leaving a tree, the parity of the sum of the degrees
 is odd. This observation is used in a
  paper on graph sketching  \cite{sketching} and
  a paper on sequential dynamic graph connectivity data structures \cite{kapron-2013}.
It is not clear how to adapt the techniques of \cite{sketching} to the distributed setting. Those of \cite{kapron-2013} were adapted to a distributed version \cite{mountjoy}, however, it was not impromptu (required keeping supplemental storage between updates) and did not address an MST (and was much more complicated that the repair algorithm presented here).

\medskip

\noindent
{\it MST and ST construction:} The complexity of the first distributed MST construction algorithm was not analyzed \cite{first-mst}.
Following the seminal paper of \cite{GHS} mentioned above, Awerbuch improved the time complexity to $O(n)$ \cite{AwerbuchMST}, retaining the same message complexity of $O(m+ n \log n)$.
Distributed algorithms that are faster (when the diameter of the network or the diameter of the MST are smaller) do exist
\cite{kdom,kutten-peleg-garay, elkin-mst}.
However, their message and memory space
complexities are much higher.
 
\medskip
\noindent
{\it Simultaneous edge changes:}
As opposed to the previous $o(m)$ (but non impromptu) repair algorithms \cite{awerbuch-2008,porat},
  ours has not been extended to deal with multiple updates at a time, though we believe it can be. \footnote{The idea behind such an extension would be, essentially, to use the algorithm of Awerbuch et. al  of 2008, but replace their method of finding replacement edges with the method used here.}

\smallskip
\noindent
{\bf Definitions and Organization}\\
\smallskip
\noindent
{\it Definitions:} An edge $\{u,v\}$'s {\it edge number} is
the concatenation of the unique ID's of the edge's endpoints, smallest first.
We
create unique weights
(as in \cite{GHS}) by
concatenating the weight to the front of its edge number. For any tree, {\it maxID(T)}, {\it maxEdgeNum(T)}, and  {\it maxWt(T)} denote the maximum ID of any node in $T$, the maximum edge number, and maximum weight edge, resp. of any node in $T$. $T$ is omitted where it is understood from context. Let $[j,k]$ denote the set of integers $\{j,j+1,...,,k\}$ and $\lg n$ denote $\log_2 n$.

\smallskip

\noindent
{\it Organization:}  The functions $\TestOut$ and $\HPTestOut$ are described in Section \ref{s:oddhash}.  Section \ref{s:findmin} describes $\Find$, an algorithm for dynamic MST and an algorithm for constructing an MST \ref{ss:buildMST}. Section \ref{s:findany} describes $\anyfind$ and reduces the complexity for construction and repair for ST. The Appendix contains an extension to the case where edge weight may be superpolynomial in $n$,  and some deferred  proof.

\section{\TestOut} \label{s:oddhash}
In this section, we describe $\TestOut$ and $\HPTestOut$.
\subsection{Random odd hash functions and TestOut}  \label{sub:testout}
As a method to sample edges, we use the concept of an {\em odd\/} hash
functions:
We say that a random hash function $h:[1,m]\fct
\{0,1\}$ is {\em $\eps$-odd}, if for any given non-empty set
$S\subseteq [1,m]$, there are an odd number of elements in $S$ which
hash to 1 with probability $\eps$, that is,
\begin{equation}\label{eq:non-zero}
\Pr_h\left[\sum_{x\in S} h(x)=1\mod 2 \right]\geq\eps.
\end{equation}

An odd hash function is
a type of ``distinguisher" described in
 \cite{thorup-sample}; we use their construction here\footnote{\cite{thorup-sample} was
originally inspired by the developments in the current paper.}. Let $m\leq 2^w$. We pick a
uniform odd multiplier  $a$ from $[1,2^w]$ and a uniform threshold $t\in [1,2^w]$.
From these two components, we define $h:[1,2^w]\fct \{0,1\}$ as
\begin{eqnarray*}
h(x)& = & 1 \hbox{ if  }(ax\bmod 2^w)\leq t\\
       & = & 0  \hbox{ otherwise.}
 \end{eqnarray*}
The above is particularly efficient if $w\in \{8,32,64\}$ in a programming language like C, for there the mod-operation
comes for free as part of an integer multiplication which automatically
discards overflow beyond the $w$ bits.
From  \cite{thorup-sample} we see
that $h$ is an $(1/8)$-odd hash function.

Let $h:[1,maxEdgeNum] \rightarrow \{0,1\}$ be an odd hash function.
We show how to compute $TestOut$.
Let $E(v)$ denote the edge numbers of edges incident to node $v$.  Let $Cut(T, V\setminus T)$ denote the set of edges with exactly one endpoint in $T$.
To test with constant probability whether there exists any edge leaving a tree $T$,
each  node $v \in T_x$ with $E(v) \neq \emptyset$ computes
  $$ \sum_{e\in E(v)} h(e) \mod 2$$
  \vspace{2mm}
locally.  If $E(v)=\emptyset$, then 0 is returned.
These values are aggregated over the nodes in $T$  to compute
$$\sum_{v\in V(T)} \sum_{e\in E(v)} h(e) \mod 2=
\sum_{e \in Cut(T, V\setminus T)} h(e) \mod 2$$

\medskip

$\TestOut(x)$ can be done with one broadcast of $h$ from node $x$ and one 1-bit echo. First
$x$  broadcasts $h$ in one message. The leaves return the parity of their sum to their parent;
the parent passes to its own parent the sum mod 2 of its children and of its own sums.

\medskip

 $\TestOut(x,j, k )$ checks if there is any edge
leaving $T_x$ whose weight is in a given interval $[j,k]$.  To do so, in each local computation at node $v$, the definitions above for
$\TestOut(x)$ are changed so that

 \[\sum_{e\in E(v)} h(e) \mod 2 \;\;\;\;\;\;\;\;\;\;\;\;\;\;\;\; {\rm is \; replaced \; by}  \;\;\;\;\;\;\;\;\;\;\;\;\;\;\;\; \sum_{e\in E(v) \wedge weight(e)\in [j,k])} h(e) \mod 2\]

\subsection{High probability \TestOut} \label{s:testout}

 $\TestOut$ achieves a constant probability of correctness if the set is non-empty and
is always correct if it is empty.  The probability can be amplified to high probability by $O(\log n)$ independent parallel repetitions.
However this would require $O(\log^2 n)$ bits. Alternatively, deterministic amplification methods can bring down the randomness to $O(\log n)$. However, this is not simple and would require each node to construct a portion of an averaging sampler.

Instead, we take advantage of the type of set we are looking at and introduce a high probability version of $\TestOut.$
W.h.p., $\HPTestOut(x)$ outputs 1 if the tree $T_x$ has any leaving edge. If there is no such edge, it always returns 0.

\medskip

\noindent
For a vertex $u$, let $E^\uparrow (u)=\{(u,v)\in E\}$ and
$E^\downarrow (u)=\{(v,u) \in E\}$. For the tree $T$, $E^\uparrow
(T)=\bigcup_{u\in T} E^\uparrow(u)$ and $E^\downarrow
(T)=\bigcup_{u\in T} E^\downarrow(u)$.
\begin{observation}\label{obs:out-test}
There is an edge $\{u,v\}\in E$ with only one endpoint in $T$ if and only
if $E^\uparrow (T)\neq E^\downarrow (T)$.
\end{observation}

Thus, to implement $\HPTestOut$, we need only test if $E^\uparrow
(T)\neq E^\downarrow (T)$. To test set equality efficiently, we use a method
from \cite{blum-kannen} based on the Schwartz-Zippel \cite{schwartz} polynomial
identity testing.
Let $B$ be the number of edges incident to nodes in $T$. To achieve probability of error $\epsilon(n)$, it suffices to use any prime $p >\max\{maxEdgeNum(T), B/\epsilon(n)\}$, with $|p| \leq w$, the maximum message size. We note that if $w$ is  sufficiently large and known to all nodes, we may take $p$ to be the maximum prime $p$  with $|p|<w$ or have some other predetermined value for $p$.
For $p$ and an edge set $D$, we define
a polynomial over $\Zp$ by
\[\:\:\:\:\:\:\:\:\:\:\:\:\:\:\:\:\:\:\:\:\:\:\:\:\:\:\:\:\:\:\:\:\:\:\:\:\:\:\:\:\:\:\:\:\cP(D)(z)=\prod_{e\in D}  (z-edge\_number(e)) \mod p.\]

\begin{equation}\label{eq:test}
{\rm From \cite{blum-kannen}:} \;\;\;\;\;\;\;\;\;\;\;\; \Pr_{\alpha\in \Zp} [\cP(E^\uparrow (T))(\alpha)=\cP(E^\downarrow (T))(\alpha) ] < \epsilon(n).
\end{equation}

\noindent
$\HPTestOut(x)$:\{assumes $x$ knows $\epsilon(n)$\}\\

\noindent
0)  If $p$ is not known by all nodes, then $x$ initiates a $Broadcast-and-echo$ to find $maxEdgeNum$, $B$ (by summing up the degrees of nodes in $T$), and using these, determines $p.$ \\ \\
1) $x$ initiates a $Broadcast-and-echo$ in which a randomly selected $\alpha \in \Zp$ (and $p$ if necessary) is passed to all nodes in the tree in the broadcast phase.  Each node $y$
locally  computes $Local^\uparrow (y)= \cP(E^\uparrow (y))(\alpha)$ and $ Local^\downarrow (y)=\cP(E^\downarrow (y))(\alpha)$. Upon receiving
$\cP(E^\uparrow (T_z))(\alpha)$ and $ \cP(E^\downarrow (T_z))(\alpha)$ from each of its children $z$,
each node computes and sends to its parent $$ \;\;\;\;\;\;\;\;\;\;\;\;\;\;\;\;\;\;\;\; \cP(E^\uparrow (T(y))(\alpha) =  Local^\uparrow(y) *\prod_{z~child ~of~ y} \cP(E^\uparrow (T_{z}))(\alpha)$$
  $${\rm and} \;\;\;\;\;\;\;\;\;\;\;\;\;\;\;\; \cP(E^\downarrow (T(y))(\alpha)=  Local^\downarrow(y) *\prod_{z~child ~of~ y} \cP(E^\downarrow (T_{z}))(\alpha) $$
3) $x$ determines there is an edge leaving $T$ iff  $\cP(E^\uparrow (T))(\alpha)\neq \cP(E^\downarrow (T))(\alpha)$.

\medskip

\noindent
{\it Analysis:}  As all computations are over $\Zp$, the number of messages sent is $\leq 4|T|$  with each containing $|p|=O(\log (maxEdgeNum+B))=O(\log n)$ bits.

\medskip

\noindent
 $\HPTestOut(x,j, k )$ checks if there is any edge
leaving $T_x$ whose weight is in a given interval $[j,k]$.  To do so, in each local computation at node $v$, the definitions above for
$\HPTestOut(x)$ are changed so that  $E^\uparrow (u)=\{(u,v)\in E\}$ and $E^\downarrow (u)=\{(v,u)\in E\}$ are replaced by  $E^\uparrow (u)=\{(u,v)\in E \wedge weight (u,v) \in [j,k] \}$ and $ E^\downarrow (u)=\{(v,u)\in E \wedge weight(v,u) \in [j,k]\}$.

\section{MST Build and Repair}\label{sec:Find} \label{s:findmin} \label{sec:findmin}

We present a simple method to find the lightest leaving edge using a $w$-wise search on the edge weights. This yields a method using $O(\log n/\log \log n)$
broadcast-and-echoes
with $w=O(\log n )$ bit messages when  the
weight of
every edge is polynomial.
 In the appendix, we give a more complicated method which uses $O(\log n/\log \log n)$ for superpolynomial edge weights of size $u$ which assumes wordsize $O(\log (n+u))$.

 \subsection{Integer edge weights  of polynomial size}
Since $\TestOut$ uses a single bit ``echo",  a single $broadcast-and-echo$ can test $w=O(\log n)$ subranges concurrently, as the same hash function can be used for each of the parallel $\TestOut$'s, while the single bit responses for each subrange
$\TestOut$'s are returned concurrently in one word.  The smallest subrange testing positive becomes the next range of edge weights to be tested. Before narrowing the range, the result is verified  w.h.p. using $ \HPTestOut$.

Below we describe $\Find$ and $\Find$-C. $\Find$-C is like $\Find$ except that the number of repetitions of the loop in the algorithm is limited to double the expected number, $O(\log n /\log\log n)$, rather than $O(\log n)$ in the worst case.  Let $q$ be the probability that $\TestOut$ succeeds.  We assume for any constant $c$, $x$ knows  a polynomial bound on the network size $n$ in order to set an error parameter for $\HPTestOut$,  $\epsilon(n) \leq n^{-c-1}$ such that $\epsilon(n)^{-1}$ is polynomial in $n$ and a bound for $Count$ for $\Find$ which exceeds $(c/q) \lg n$ and is $O(\log n)$.

\medskip

\noindent
{\bf $\Find(x)$ [$\Find-C$]} \{{\it finds minimum cost edge in $(T_x, V\setminus T_x)$}\}
\vspace{-2mm}
\begin{enumerate}[noitemsep]


\item

$Count \leftarrow 0.$
\item
 $x$ determines $maxWt(T_x)$ and $maxEdgeNum(T_x)$ through one broadcast-and-echo and computes $\epsilon(n)$.

\item
 $x$ sets $j \leftarrow 1$; $k\leftarrow maxWt(T_x)$

\item $x$ broadcasts an odd hash function $f:[1,maxEdgeNum(T_x)] \rightarrow \{0,1\}$ and also $j$ and $k$.\label{line:newhash}

\item In parallel for $i=0,1, 2,..., w -1$:\\
 set $j_i=j+i\lceil{(k-j)/w}\rceil$ and $k_i=j+(i+1)\lceil{k-j)/w}\rceil -1$,\\
return word in which $i^{th}$ bit is the ``echo" of $\TestOut(x,j_i, k_i)$

\item Upon receiving the echo, $x$ determines the index ${min}=\min\{i ~|~\TestOut(x,j_i,k_i) =1)$ and \\
 initiates  $TestLow =\HPTestOut(x,0,j_{min}-1)$ and $TestInterval=\HPTestOut(x,j_{min},k_{min})$.

\item Upon receiving results,
\begin{enumerate}[noitemsep]
\item
if $TestLow=0$ and $TestHigh=1$ \\
and if $j_{min} < k_{min}$ then $x$ sets $j=j_{min}$ and $k=k_{min}+1$; else  if $j_{min}=k_{min}$  $x$ broadcasts ``stop" and returns $j_{min}$.
\item
else if both return 0, $x$ broadcasts ``stop" and returns $\emptyset$.
\end{enumerate}

\item
For $\Find$  [resp., $\Find-C$]: If $Count < (c/q)\lg n+ (c/q) (\lg maxWt(T_x)/\lg w$, [resp., $Count < (2c/q) \lg maxWt(T_x)/\lg w $], increment Count and repeat from Step \ref{line:newhash}. Else return $\emptyset$.

 \end{enumerate}

  \vspace{1mm}
 \noindent
  {\bf Proof of correctness}\\
 \vspace{-5mm}

 \begin{lemma} \label{l:find}
Let $c$ be any constant s.t. $c \geq 1$.  With probability $1-n^{-c}$, using asynchronous communication, $\Find(x)$ returns the lightest edge leaving a tree $T_x$ in expected time and messages $O(|T_x| \log n/\log \log n)$ (and worst case $O(\log n)$ time and messages.
With probability $2/3 -1/n^c$ $\Find-C(x)$ returns the lightest edge and with probability $1-n^{-c}$  it returns the lightest edge or $\emptyset$, using worst case $O(|T_x| \log n/\log \log n)$ messages and time. If there is no edge leaving the tree, both procedures always return $\emptyset$. This assumes $x$ knows an upper bound on the size $n$ of the network which is polynomial in $n$.

\end{lemma}

\begin{proof}
$\Find$ is analyzed first.
We observe that if $\HPTestOut$ is always successfully, then $\Find$ will terminate successfully after no more than
$\lg {maxWt}/\lg (w-1)$ successful executions of $\TestOut$:
Let $I=(j_i, k_i)$ be the first interval containing an edge leaving $T_x$. $\TestOut$ always returns a 0 for earlier intervals, and returns a
1 with constant probability $q=1/8$ when $I$ is tested.
  If $\TestOut$ fails to return a 1 for interval $I$, then $TestLow$ will detect a 1 and the loop is repeated; otherwise the range is successfully narrowed.
The range is narrowed no more than $\lg {maxWt}/\lg (w-1)$ times.  Each successful narrowing requires an expected $1/q$ repetitions and overall, in expectation  $(1/q)\lg {maxWt}/\lg (w-1)=O(\log n/\log n \log n)$ iterations of Steps 4-8 suffice to return the lightest edge leaving $T$ (if such exists).

We bound the probability that $\TestOut$ fails $K$ times before succeeding $\lg {maxWt}/\lg (w-1)=O(\log n/\log \log n)$ times, where
$K=(c/q) \lg n $:
This is given by a tail bound on a random variable with a binomial distribution with $K+ \lg maxWt/\lg (w-1)$ trials and constant probability $q$ of
heads (success). Using Chernoff bounds,
the probability of this type of failure is $< 1/(2n^c)$ for sufficiently large $n$.

We now bound the probability that $\HPTestOut$ fails at least once after any of these calls to $\TestOut$: With an error parameter of $ \leq n^{-c-1}$ for  $\HPTestOut$, the probability of the latter over $2K=2(c/q)\lg n$ trials is less than $1/(2n^c)$ by a union bound.

We conclude that the probability of either event occurring is less than $1/n^c$, again by a union bound. Hence, w.h.p., the range is successfully narrowed to the minimum weight edge after $2(c/q)\log n$ iterations of Steps 4-8 or, if there is no edge leaving $T$, then
Step 7(b) is executed and $\emptyset$ is returned.

 For $\Find-C$, $\TestOut$ is restricted to make only $K'=(2c/q) \lg maxWt/\lg (w-1)$ repetitions.  For $\Find-C$ to return the lightest edge,
 $\TestOut$ cannot fail more than $K'$ times before achieving $\lg maxWt/\lg (w-1)$ successes and $\HPTestOut$ cannot fail once in $2(K' +
 2\lg maxWt/\lg (w-1) $ trials. The probability that the number of $\TestOut$  trials needed exceeds the expected number by a factor of $2c/q +1$ is less than $1/3$ for $c \geq 1$, by Markov's Inequality.
 We next bound the probability that $\HPTestOut$ fails during any one of the $2K' +2\lg maxWt/\lg (w-1)$ trials and repetitions. Since the error parameter for $\HPTestOut$ is $n^{-c-1}$, the union bound over all these $<n$ trials gives a probability of this type of error of $1/n^c$. For $\HPTestOut$ to return an incorrect
 lightest edge, $\HPTestOut$ must fail at least once.
Hence, if there is an edge leaving then with probability $2/3-1/n^c$, $TestOut-C$ returns the correct edge, with probability $1-1/n^c$ it returns $\emptyset$ or the correct edge, and with probability $<1/n^c$ it returns an edge leaving the tree which is not the lightest edge.

 \end{proof}

\subsection{Impromptu repairs of MST} \label{sec:well-sep}
We now apply $\Find{}$ to the problem of repairing an MST.
Assume that the updates are
well-separated in the sense that we can complete the processing of an edge
update before the next one arrives. Before
any update, assume the network maintains a minimum spanning forest, and each node knows some polynomial (in $n$) upper bound on the size $n$\footnote{This is only required to compute with probability of error a function of $n$}.

\paragraph{Delete$(u,v)$.} When an edge $\{u,v\}$ is deleted, if $u <v$, then if  $\{u,v\}$  was in the MST, then
$u$ initiates $\Find$ in the marked subtree containing $u$, $T_u$.  If
$\Find$ returns $\emptyset$, it means that $\{u,v\}$ was a bridge, and we
are done. Otherwise $\Find$ returns an edge $\{u',v'\}$.  Then $u$
broadcasts that $\{u',v'\}$ should be added to the minimum spanning
forest, and $u'$ forwards this message to $v'$. Both $u'$ and $v'$ mark the edge $\{u',v\}$. The
bottleneck of the complexity is the call $\Find(u)$ which uses
$O(n_u\log n)$ messages
for $n_u\leq n$
nodes in $T_u$.

\paragraph{Insert$(u,v)$.} When an edge $\{u,v\}$ is inserted, and $u <v$,
$u$ determines if its tree $T_u$ in the MST contains $v$ and if so, it determines
the heaviest edge $e$ on the path from $u$ to $v$. This is easily
done by a broadcast-and-echo from $u$. If $e$ is  heavier than $\{u,v\}$,
 $\{u,v\}$ is included in the minimum spanning forest, and $u$ broadcasts that $e$
should be removed from the MST. A constant number of broadcast-and-echoes are used, for a total
number of messages which is proportional to the size of $T_u$.

\smallskip

The
analysis of these operations follow from Lemma \ref{l:find}. With the extension of $\Find$ to superpolynomial edge weights given in the Appendix, the proof of Theorem \ref{t:repair} follows.

\subsection{Building an MST} \label{ss:buildMST}
In a synchronous network, building an MST from scratch is a
straightforward application of $\Find$.
Recall (the Introduction) that we assume that every node knows $n$\footnote{up to some constant factor. } and the list of the edge weights of its incident edges and that the edge weights of all edges are distinct.

The goal is for each node to mark a subset of its neighbors so that the resulting marked edges form an MST. The algorithm is an implementation of Bor\r{u}vka's parallel algorithm for constructing an MST. During the execution, the nodes are partitioned into fragments, each a connected component of the final MST. (Initially, each node is a singleton fragment).
At each round, in parallel, a minimum weight edge incident to each non-maximal tree (fragment) is found by a search started by the {\em fragment leader}.

{\it Electing a fragment leader} is straightforward and is similar to a broadcast-and-echo and ideas in \cite{KorachRotem}: Since this is a synchronous network, all nodes know when an iteration starts and thus when to start the leader election. Moreover, every leaf of a fragment knows it is a leaf and so should start.
Each leaf acts as if it has just received a broadcast message initiated by the leader (though the leader is not known yet). That is, the leaf sends an echo message to its (only) tree neighbor - thus designating that neighbor as its parent. As in broadcast-and-echo, every internal node who received an echo from all its neighbors but one, sends an echo to that last one. It is then easy to see that
either the tree has one median or two. In the first case, the echoes converge to that median. Let us elect this one the leader. In the second case, there are two neighboring medians. Let the one with the higher identity be the leader.

Let $maxTimeMST(n)$ be the maximum amount of time needed to carry out Steps (a)--(c) in a tree of size $n$. We assume a global clock with value $time$. Let $C$ be the (constant) probability that $\Find-C$ returns the minimum edge incident to a tree, if there is one. Let $c$  in the algorithm below be the desired (constant) parameter, such that the probability of success of the Build MST algorithm should be $1-n^{-c}$.

\medskip
\noindent
{\bf Build MST} \{{\it executed by every node $x$}\}
\vspace{-2mm}
\begin{enumerate}[noitemsep]
\item  $time \leftarrow 0$
\item
For $i=1$ to $(40c/C)\lceil{\lg n}\rceil$:
\vspace{-1mm}
\begin{enumerate}[noitemsep]
\item
Elect a leader in $T_x$

\item
If $x=leader$ then $x$ initiates $\Find-C$; else $x$ participates in $\Find-C$.
\item
If $x$ is an endpoint of the edge $\{x,y\}$ which has been returned by $\Find-C$, $x$ sends Add\_Edge message to $y$ across $\{x,y\}$.

\item While $time < i*maxTimeMST(n)$  wait; while waiting, if any Add\_Edge message is received over an edge, mark that edge.\label{wait}
\end{enumerate}
\end{enumerate}

\begin{lemma}
Let $c$ be any constant, $c\geq 1$. With probability $1-n^{-c}$, $Build\_ MST$ constructs an MST in time and message complexity $O(n \log^2 n/\log \log n)$.
\end{lemma}
\begin{proof}
We call each for-loop a {\it phase}. At the start of each phase there is a forest of trees consisting of all marked tree edges. If there is an edge leaving the tree in the the graph of all edges, it is a {\it non-maximal tree}.
The variable $maxTime(n)$ is set so that every node enters phase $i$ at the same time, after completing phases $j<i$.
We first show:

\smallskip

\noindent
{\it Claim 1:} After $(16c/C) \lg n $ phases,  the number of non-maximal trees is no greater than $(8c/C) \lg n$  with probability $1-1/n^c$.

\smallskip

\noindent
{\it Proof of Claim 1:}Fix a phase in which the number $F$ of non-maximal trees is greater than
$(8c/C)\lg n$.
  For each $T_j$ which is non-maximal at the start of
the phase, let $X_j=1$ if the execution of $\Find$ returns the minimum
weight edge incident to $T_j$ and 0 otherwise. Then each $X_j$ is an
independent random variable with constant probability $C$ of
success. Using Chernoff bounds we can see that at least $C/2$ of these
$\Find$'s succeed with high probability: $Pr(\sum X_j < (C/2) F
) < exp(-((1/2)^2 C(8c/C)\lg n))/2)< 1/n^{c}$.
Thus, w.h.p., the
number of non-maximal components is reduced by a fraction of $C/2$ in
each phase, until fewer than $(8c/C) \lg n$ non-maximal trees
remain. This requires no more than $\lg n /\lg (C/2)=O(\log n)$
phases.  Each phase executed by a tree of size $s$ uses a number of
messages $O(s\log n/\log \log n)$ or $O(n \log n/\log \log n)$ over
all non-maximal trees.  Now we show:

\smallskip

\noindent
{\it Claim 2:} If there are $c' \lg n$ non-maximal components to start, then after  $(2c' + 8c)/C) \lg n$ more phases, there are no non-maximal trees, with probability $1-1/n^c$.
\smallskip

\noindent
{\it Proof of Claim 2:} For each phase that there is a non-maximal tree which successfully runs $\Find-C$, the number of non-maximal components is reduced by one. We call a phase successful if there is at least one successful run of $\Find-C$. For any $c$, after $((2c'+8c)/C) \lg n$ phases with at least one  execution of $\Find$ each, at least $c'\lg n -1$ $\Find$'s will be successful with probability at least $1-1/n^{c}$.

\smallskip
Putting these claims together:
Let $c'=(8c/C)$. Then $(16 + 24)(c/C) \lg n$ phases suffice to reduce the number of non-maximal trees from $n$ to 0, with probability $ \geq 1-n^c$.
We conclude the proof of the lemma by observing that $O(\log n)$ phases require a total of $O(n \log^2 n/\log \log n)$ messages and time.
\end{proof}

\section{Unweighted edges}

We now present analogous results for unweighted graphs using less costly, somewhat different techniques.
\subsection{Find any edge leaving a tree}
\label{s:findany}
$\anyfind$, presented below, uses an expected constant
number of broadcast-and-echoes, to find any edge leaving $T_x$.
Thus in expectation, we save a factor $\log n/\log \log n$ in the asymptotic cost of $\Find$.

The procedure starts with $\HPTestOut$ to determine if there is an edge in the cut w.h.p.
 If $\HPTestOut$ returns 1,  a routine to find such an edge with a constant probability of success is run repeatedly until
   such an edge is found,
   yielding a constant expected time and message procedure.
To achieve a probability of error $n^{-c}$  in the running times claimed, we assume $x$ knows an $\epsilon(n) < 1/(2n^c)$ where $\epsilon^{-1}(n) $ is polynomial in $n$. $T$ below is $T_x$. We let $[r]$ denote the set $\{0,1,...,r-1\}$.
\medskip

\noindent
$\anyfind(x)$

\vspace{-4mm}
\begin{enumerate}[noitemsep]
\item
$Count \leftarrow 0$.

\item
$x$ initiates $\HPTestOut$ in $T$ with error parameter $\epsilon(n)$. If $\neg (\HPTestOut)$ then return $\emptyset$.

\item
 Determine the identity of an edge as follows:\\[2mm]
a) $x$ broadcasts a random pairwise independent hash function $h: [1,maxEdgeNum(T)]  \rightarrow [r]$ where $r$ is a power of 2 $>$ sum of degrees of nodes in $T$.\\[2mm]
b) Each node $y$ hashes the edge numbers of its incident edges using $h$, and computes  the vector $\vec{h}(y)$ s.t. ${h_i}(y)$ is the parity of the set of its incident edges whose edge numbers hash to values in  $[ 2^i]$  for $i=1,..., \lg r$.  If $y$ has no incident edges then
$\vec{h}(y)=\vec{0}$. \\[2mm]
c)  The vector $\vec{h(T)}=\XOR_{y \in T} \vec{h}(y)$ is computed up the tree, in the broadcast-and-echo return to $x$.  Then $x$ broadcasts $min=
\min \{i ~|~ h_i(T)=1\}$. \\[2mm]
 d)  Let $E(x)$ be the set of edge numbers of edges incident to $x$.
 Each node $x$ computes $w(x)=\XOR \{e ~|~ \{e \in E(x) \wedge h(e) < 2^{min} \} $ and
$w(T)=\XOR_{x \in T} w(x)$ is computed up the tree in the broadcast echo and returned to $x$. \\
\{{\it If there is exactly one edge leaving $T$ with $h(e)<2^{min}$,
then $w(T)$ is its edge number.}\}\\
\item
Test: $x$ broadcasts $w(T)$ to obtain $Sum=$ the number of endpoints in $T$ incident to the edge given by  $w(T)$. The test succeeds iff $Sum=1$.

\item
If  Test succeeds, return $w(T)$ else \\
for $\TestOut-C$, return $\emptyset$; \\
for $\TestOut$, if $Count\geq 16 \ln( \epsilon^{-1}(n))$ then return $\emptyset$; else increment $Count$ and repeat steps 3-5. \\

\end{enumerate}

 \vspace{-6mm}
 \noindent
  {\bf Proof of correctness}\\
   \vspace{-2mm}
  Let $h$ a 2-independent function from a universe $U$ into
$[2^\ell]$ for some $\ell \geq 2$. Let $W\subseteq U$ s.t.
$0<  |W|<2^{\ell-1}$.
\begin{lemma}  \label{l:findany} With probability $1/16$,
if $|W|>0$ then there is an integer $j$
such that exactly one $w\in W$ hashes to a value in $[2^j]$.
\end{lemma}
\begin{proof} We prove the statement of the lemma for $j=\ell-\ceil{\lg |X|}-1$.
Then $1/(4|W|)<2^j/2^\ell\leq 1/|W|$.
Now
\vspace{-5mm}
\begin{align*}
\Pr_{h}&\left[\exists !\, w\in W: h(w)\in [2^j]\right]\\[-1mm]
&= \sum_{w\in W} \left(\Pr_{h}\left[h(w)\in [2^j]
\wedge \forall w'\in W\setminus\{x\}: h(w')\not\in [2^j]\right]\right)\\[1mm]
&= \sum_{w\in W} \left(\Pr_{h}\left[h(w)\in [2^j]\right]
\,\Pr_{h}\left[\forall w'\in W\setminus\{w\}:h(w')\not\in [2^j] \mid h(w)\in [2^j]\right]\right)\\[1mm]
&\geq \sum_{w\in W} \left(\Pr_{h}\left[h(w)\in [2^j]\right]
\left(1-
\sum_{w'\in W\setminus\{w\}}\Pr_{h}\left[h(w')\in [2^j]\mid h(w)\in [2^j]\right]\right)\right)\\
&=\sum_{w\in W} \left(\Pr_{h}\left[h(w)\in [2^j]\right]
\left(1-
\sum_{w'\in W\setminus\{x\}}\Pr_{h}\left[h(w')\in [2^j]\right]\right)\right) \hbox{by 2-wise independence} \\
&= |W|\left(2^j/2^\ell\cdot (1-(|W|-1)2^j/2^\ell)\right) \;\;>\;\;|W|/(4|X|)(1-|W|/(2|W|)=1/16.
\end{align*}
\end{proof}

\begin{lemma} \label{l:anyfind2} If there is an edge no leaving $T_x$, then $\anyfind (x)$ and $\anyfind -C(x)$ return $\emptyset$. Otherwise,
\begin{itemize}[noitemsep]
\item
$\anyfind (x)$ returns an edge leaving $T_x$ w.h.p.
 It uses expected time and messages $O(n)$;  and
\item
$\anyfind -C(x)$ returns an edge leaving $T_x$  with probability at least $1/16$, else it returns $\emptyset$. It uses worst case time and messages $O(n)$.
\end{itemize}
\end{lemma}

\begin{proof}
Let $W$ be the set of edge numbers of edges  leaving $T$. If $|W| =0$, then $\HPTestOut$ returns $\emptyset$ and $x$ returns $\emptyset$. If $|W|>0$ then with probability $\geq 1-1/(2n)^c$ $\HPTestOut$ succeeds and $x$ continues to Step 3.
Given $x$ goes on to Step 3,  by Lemma \ref{l:findany}, with probability at least $1/16$, there is a $j$ such that
exactly one edge with distinct edge number $e$ in $W$ hashes to a value in $[2^j]$ (``Event A"). Because all edges incident to $T$ which are not leaving $T$ have both endpoints in $T$, their edge numbers $\XOR$ to $\vec{0}$, when summed over $T$.  Hence, when Event A occurs,
 $\XOR_{y \in T} h_j (y)= \XOR_{e' \in W}h_j(e') =h_j(e)=1$,
so $min\leq j$.  However, $\XOR_{y\in T} h_{min} (y)=1$ implies that
there is at least one edge number in $W$ which hashes to $[2^{min}]\subseteq [2^j]$, so we
conclude that there is exactly one such edge number $e\in W$ hashing to $[2^{min}]$.

When Event A occurs, in Step 4, $w(T)=e$ and Test succeeds in Step 5 and an edge leaving $T$ is returned. Thus the probability of success of $\anyfind-C$ is the probability that  $\HPTestOut$ succeeds, followed by Event A which is $\geq 1/16-1/(2n^c)$. In $\anyfind$, if $\HPTestOut$ succeeds, then Steps 3-5 are repeated up to  $16\ln(\epsilon^{-1}(n))=16\ln(2 n^c)$ times
until they succeed. The probability of failure of
all these repetitions is $\leq (1-1/16)^{16\ln(2n^c)} <1/(2n^c)$. The total probability of failure is therefore no more than this probability plus
the probability of failure of $\HPTestOut$ for a total probability of failure $\leq 1/n^c$. The expected number of repetitions of Steps 3-5 until success is 16 (and the worst case is $O(\log n)$).

Since a single run of Steps 1-5 requires $O(n)$ time and messages, the lemma follows.

\end{proof}

\subsection{Building an ST}
This algorithm is obtained by modifying the algorithm for  building of the MST. Two modifications are necessary.
The first is the substituting of $\anyfind-C$ for $\Find-C$ in each. The replacement of the $O(n \log n)$ $\Find-C$ by
$\anyfind-C$ reduces the asymptotic costs by a factor of $\log n/\log \log n $.

The second is more subtle. In $\Find$, when all fragments pick minimum weight edges leaving them, all of them are MST edges. This is because the weights of the edges are distinct, there is only one minimum weight edge leaving any fragment. Moreover, such an edge must be in the MST. When $k$ fragments of the unweighted graph pick edges leaving them, it is possible that $k$ distinct edges are picked and (at most one) cycle is formed. by ``potential" tree edges. This needs to be detected before the next phase begins.

If all nodes run the leader election
algorithm described in Section
4.3, the nodes on the cycle will be exactly the set of nodes which
fail to hear from all but
two
of their neighbors.
 After the maximum
time needed for leader election, these nodes will be aware they are on
a cycle.
 Moreover, they know their neighbors in the cycle, since they have
not heard from them.
Each node randomly picks one of the two edges
incident to it in the cycle to exclude and sends a message along that
edge to its other endpoint. If some edge is picked by both its
neighbors, then this edge is unmarked, i.e., not added to the
tree. Leader election is again run to test if there is a cycle.  If
there still is a cycle, all of the edges in the cycle are unmarked and
not included as tree edges in the next phase. \\

The analysis of this algorithm appears in Appendix \ref{app:st}. Intuitively, beyond the analysis of Build MST, it shows that an edge is likely to unmark (breaking the cycle) in high probability. Note that at most of half of the chosen outgoing edges are unmarked, so ``enough'' mergers still occur. This ensures progress. If the cycle is small and not edge is marked, then the whole cycle is removed. Since the (removed) cycle is small, still ``enough'' mergers occur. This establishes the following lemma.

\begin{lemma}
\label{lem:ST}

Let $c$ be any constant, $c\geq 1$. With probability $1-n^{-c}$, there is an algorithm which constructs an ST in time and message complexity $O(n \log n)$.
\end{lemma}

\subsection{Repairing an ST}

This is a straightforward adaptation of the methods used for repairing an MST, except that $\anyfind$ is used in place of $\Find$
and a factor of $\log n/\log\log n$ is saved from the asymptotic cost.

\section{Open problems and conclusions}

We have adapted a technique from streaming and dynamic sequential graphs to find a surprising result, that the problem of constructing a broadcast tree can be done with $O(n \log n)$ messages (and time)
in the CONGEST model w.h.p., a problem believed to have a lower bound of $\Omega(m)$ on the number of message for 25 years or more.
(In a model allowing much longer messages, it was known how to avoid sending messages over some edges \cite{KorachKMoran}; intuitively, \cite{vainish} showed that
for each such avoided edge, the identity of one of its endpoints needs to be delivered uncompressed to the other endpoint; we have shown that those identities could be compressed significantly if the ID space is of a reasonable size, up to even exponential in $n$).
 We also have shown a very simple way to repair ST's and MST's in
$O(n)$ and $O(n \log n/\log\log n)$ expected time and messages; previously it was suggested that reducing the message complexity to $o(m)$,
 requires auxiliary information to be stored between updates. By avoiding the need to (store and) distribute  auxiliary information, we also manage to make the $o(m)$ message complexity worst case rather than just amortized as in previous papers.
Can these yield practical methods for real dynamic networks?

A number of interesting theoretic problems remain. For ST and MST construction,
  Can $ST$ be constructed by a deterministic or Las Vegas algorithm in $o(m)$ messages
   in the $K_1$ model?
 What kind of bounds need the nodes know of $n$?
Can these results be made to work in the asynchronous model of communication?
Is it possible to form an ST in time less than $O(n \log n) $ with $o(m)$ messages?
Finally, are $O(n \log n/\log\log n)$ messages required for $\Find$ or can this be pushed closer to the cost of $\anyfind$?

\medskip
\noindent
{\bf Acknowledgment:}
We would like to thank Moni Naor, Gopal Pandurangan, and Ely Porat for useful comments.

\appendix
 
\section{Accommodating superpolynomial sized edge weights}

Suppose the maximum edge weight has $w$ bits where $w$ is the message size.  We show that $O(\log n/\log \log n)$ broadcast-and-echoes suffice. In the previous subsection, the $\log n$-wise "pivots" were chosen obliviously. Here,
we use pivots based on randomly chosen edges.

Let $d$ be the total number of endpoints of nontree edges incident to the tree. Let $k=\sqrt{\log n/\log\log n}$. The $\Sample(p)$ routine described below returns $r$ sequences of $w/k$ bits from randomly sampled  edges with prefix $p$.
These edges are nontree edges with one or two endpoints in the tree, such that each non-tree edge with prefix $p$ incident to $T_x$ is picked with probability $1/m$ or $2/m$ where $m$ is the  total number of such edges.\\

\noindent
{\bf $\Find(x)$} \{{\it finds minimum cost edge in $(T_x, V\setminus T_x)$}\}
\begin{enumerate}[noitemsep]

\item
 $x$ sets $j \leftarrow 1$; $k\leftarrow w/k$, $P=\emptyset$,  and announces start,  and sends ``start"  to initiate with $\TestOut(x,j,k)$.

 \item  If $\HPTestOut(x, j, k)=0$. $x$ broadcasts ``stop" and returns $\emptyset$,

 \item while $j<k$ repeat:

\{Loop 1\}
 \item $x$ broadcasts one $O(\log n)$ bit  low probability odd hash function $f$ and one high probability hash function $F$.

 \item Run $\Sample( j,k )$\\

\item  In parallel for $i=0,1, 2,..., w/r$, run $\TestOut(x, p \cdot j_i \cdot \bar{0}, p\cdot j_{i+1} \cdot{\bar{0}})$ using $f$ on each interval where $j_0=j$ and  $j_{k+1}=k $, and other $j_i$'s are given by $\Sample(p)$. Assume they are ordered  by value.\\
\{End Loop 1\}

\item Let $min$ be the minimum $i$ s.t. $\TestOut(x,j_i p \cdot j_i, p\cdot j_{i+1}) =1$.\\
 \{Verify that there are no edges with weights with lower prefixes which leave the  tree by testing: \}\\
  $\HPTestOut (x,j_i p \cdot j_i, p\cdot j_{i+1}-1)=1$. If there are rerun
the previous step to recompute the minimum.

\old{
 \item If  $min=0$  then \{test  if  $p\cdot j_1$ is a prefix of the lowest weight leaving the tree\} \\
 if $\HPTestOut(x, p \cdot j_0 \cdot \bar{0}, p\cdot j_{1} -1 \cdot{\bar{111}})=0$ and $\HPTestOut(x, p \cdot j_1 \cdot \bar{0}, p\cdot j_{1} \cdot{\bar{1}})=1$ then set $p = p \cdot j_1$.}
 \item
 \{continue to look for an extension of $p$ or a single edge\}\\

If $j_{min} =j_{min+1}$ extend $p$ to $p\cdot j_min$, set  $j$ to $\{0\}^{w/r}$ and $k=\{1\}^{w/k}$.
\item
Else set $j $ to $j_{min}$ and $k $ to $j_{min+1}$; broadcast these values.

\item
return the edge given by augmented weight $j$.

\end{enumerate}
\old{
\begin{enumerate}[noitemsep]
\item $\ell \leftarrow \lfloor(k+j)/2 \rfloor$.
\item Let $min$ be the minimum $i$ s.t. $\TestOut(x,j_i,k_i) =1$ then $x$ broadcasts $min$ and sets $j=j_i$ and $k=k_{i}$ and repeats.
\end{enumerate}
\item
Return the edge with augmented weight $j$.
 
\end{enumerate}}

Let  $m_y$ be the number of nontree edges incident to node $y$ whose weights have property $P$. Let $m_T =\sum_{y\in T} m_y$. Let $S$ be the multiset of such edges where an edge appears twice if both its endpoints are in $T$.

\noindent
{\bf $\Sample(j,k)$} \{returns the prefixes of $r$ edges drawn uniformly at random from $S$ whose weights are in the range $[p\cdot j, p\cdot k]$.\}\\

To implement $\Sample(p,j,k)$, $x$ initiates a broadcast-and-echo to its tree. Consider the tree rooted at $x$. On the echo: each node $y$ determines and stores $\sum_z m_z$ over all $z$ in the subtree rooted at $y$. Starting with the leaves each node $y$ passes up this sum, adding on $m_y$.

  Let each node arbitrarily order its children. Then  to sample $r$ elements, $x$ randomly determines how many samples come from itself and from its children by drawing $r$ random numbers in the range from $\{1,..,total\}$. It randomly chooses the samples from itself, and sends the
 the number of requests which fall into the range of each child's edges to the child, which repeats this procedure. This requires only $\log r$ bits.
 Another echo returns the $k$ samples in parallel as each node affixes its choices.
No more than $r$ prefixes are sent to the root in total , of size $w/r$ so that they all fit in one message.

\begin{lemma}
W.h.p., the lightest edge is in the weight interval $(p \cdot j, p\cdot k)$ or there is no such edge and the algorithm returns $\emptyset$.
\end{lemma}

\begin{proof}
Initially this is true. If there is no such edge and $\HPTestOut(j,k)$ returns 0 w.h.p. and the algorithm returns $\emptyset$.

Assume it's true at the start of the loop. Then $\TestOut$ must return a 1 for some interval or it is rerun. If it returns a 1 then
the interval tested must contain an edge leaving the tree; the lighter intervals are tested w.h.p. to confirm they have no lighter edge leaving the tree. Therefore the interval $min$ must contain the lightest edge. If $j_min$ and $j_{min+1}$ agree then $p \cdot j_{min}$ must be the prefix of the lightest edge weight.  Each high prob. test has probability of failing of $1/n^c$.
There are only a constant number per iterations, hence by a union bound they all succeed w.h.p.

\end{proof}

\begin{theorem}
In a tree whose nontree edge weights are of length $w$ bits, there is an asynchronous algorithm to find the lightest nontree edge leaving the tree in $O(\log n/\log\log n)$ expected broadcast-and-echoes with message size $w$.
\end{theorem}

\begin{proof}
Correctness follows from the lemma.  We  first examine the number of iterations of loop 1.

We first note that loop 1 terminates when $\TestOut$ succeeds in the interval containing the lightest edge. This happens with constant probability. Hence it repeats a constant number of expected times.

Consider the edges (and possible duplicates)  in $S$ ordered by weight.
With each sampling, there is a constant probability that a sample edge will be chosen which is within $m_T/r$ of the lightest edge on either side of the ordering, or there are fewer than $m_T/r$ such edges. Hence if their prefixes are different,  there is a constant probability that $S$ in the next round has size $2m_T/r$. The number of these ``successful" samplings needed to shrink  the number of such edges to less than $r$ is $log_r m_T$.
The expected number of samplings to achieve this many successful rounds is $O(log_r m_T) \leq \log (n^2) /\log r= O(\log n/\log \log n).$

On the other hand, if the prefixes are the same, then we extend the prefix another $w/k$ bits. The maximum number of these samplings is $w/(w/r)=r=\sqrt{\log n/\log\log n}$.
\end{proof}

\section{Analysis of Lemma \ref{lem:ST}}
\label{app:st}

We analyze Build ST by modifying the analysis of Build MST.
We observe that with probability at least $1-1/2^{k-1}$ for a cycle of size $k$, at least one edge is unmarked and there is no more cycle, while no more than half the edges in the cycle are unmarked.  If an edge is found in the cycle that can be unmarked then if the number of fragments would have dropped by a factor of $C/2$ in the analysis of Build ST, in Build MST it drops by at least a factor of $C/4$ and the analysis is not very different.

Suppose no edge in the cycle is found and all edges in the cycle are unmarked.
As in the analysis of Build MST, there are two cases.
The first is when there are at least $c' \lg n$ fragments:
If the cycle involves less than half the fragments which successfully found edges then the number of edges which become unmarked is less than half the
total number of edges which were marked, and again, the number of fragments drops by a factor of $C/4$ instead of $C/2$.
The case where the edge unmarking does not break the cycle but
it involves at least half the $c' \lg n$ fragments
(so the whole large cycle is removed) happens with probability less than
than $1/2^{(c'/2) \lg n -1}$.
 For $c' >2c+1$, w.h.p., this case does not happen.

The second case is when the number of fragments is less than $c' \lg n$. There is a probability of at least $C$ that at least one fragment will find an edge leaving, and given that this happens, there is a probability of at least $1/2$ that if a cycle is formed, it will be broken, so that at least one new marked edge is added during the phase with probability $C/2$. Following the analysis similar to Build MST, after $O(n \log n)$ more phases, the ST tree will be formed w.h.p.

\end{document}